\documentclass{llncs}

\usepackage[numbers]{natbib}
\usepackage[utf8]{inputenc}
\usepackage[english]{babel}
\usepackage{amsfonts}

\usepackage{amsmath,amssymb,amsthm}
\usepackage{algpseudocode}
\usepackage{algorithm}
\usepackage[symbol*]{footmisc}

\newcommand{\RAM}{$\mathit{RAM}$ }
\newcommand{\PRAM}{$\mathit{PRAM}$ }
\newcommand{\TO}{\textbf{\textit{time optimal }}}
\newcommand{\while}{$\textbf{while}$ }

\newcommand {\bs}[0]{}

\begin{document}
\title{Parallel multiple selection \\ by regular sampling \thanks{This work was supported by the Polish National Science Centre grant DEC-2012/07/B/ST6/01534 } }

\author{Krzysztof Nowicki}
\institute{
Institute of Computer Science\\
University of Wrocław\\
Wrocław, Poland\\
Email: krz.nowicki@gmail.com
}
\maketitle

\begin{abstract}
\par In this paper we present a deterministic parallel algorithm solving the multiple selection problem in congested clique model. In this problem for given set of elements $\bs{S}$ and a set of ranks $\bs{K = \{k_1, k_2, ..., k_r\}}$ we are asking for the $k_i$-th smallest element of $\bs{S}$ for $\bs{1 \leq i \leq r}$.
\par The presented algorithm is deterministic, \TO, and needs $\bs{O(\log^*_{r+1}(n))}$ communication rounds, where $\bs{n}$ is the size of the input set, and $\bs{r}$ is the size of the rank set. This algorithm may be of theoretical interest, as for $\bs{r=1}$ (classic selection problem) it gives an improvement in the asymptotic synchronization cost over previous $\bs{O(\log\log p)}$ communication rounds solution, where $\bs{p}$ is size of clique.
\end{abstract}

\section{Introduction}
Parallel algorithms are one of the answers to the question of what to do with the large amount of data encountered in today's systems. In this paper we will present another deterministic algorithm solving the multiple selection problem designed for parallel architectures similar to synchronous message passing model. In such models, each processor has some private memory and computations is peformed in rounds consisting of alternating phases of local computation and communication.

\subsection{Complexity measurements}
We will analyse two kinds of computational complexity:
\begin{itemize}
\item round complexity
\item time complexity
\end{itemize}
The round complexity is the number of rounds required to finish execution of an algorithm. The time complexity of one round is the minimal number of operations after each comptational unit finishes local computation in this round. The time complexity of an algorithm is simply sum of time complexities over all rounds.

\par We will consider an algorithm solving some given problem \TO if its time complexity is $O(\frac{\mathit{OPT}}{p})$, where $\mathit{OPT}$ is time complexity of an optimal algorithm solving this problem on a sequential \RAM machine.

\subsection{About multiple selection problem}
\par In the multiple selection problem $(S,K)$ for given set of elements $S$ of size $n$ and a set of ranks $K = \{k_1, k_2, ..., k_r\}$ we are asking for the $k_i$-th smallest element of $S$ for $1 \leq i \leq r$.
\par The multiple selection problem has direct application in statistics, more specifically we solve this problem to find quantiles. Apart of that, for $r = n$ it is classic sorting problem and for $r=1$ the selection problem, which are among basic problems of computer science.
\par The multiple selection as generalization of sorting and selection might be also interesting in context of parallel architectures with local memory, from purely theoretical point of view. The sorting problem has a deterministic \TO solution requiring $O(1)$ communication rounds. For selection problem we do not know such solutions nor any non trivial $\omega(1)$ lower bounds on required number of communication rounds. Known algorithms are either randomised, non \TO or require $\omega(1)$ communication rounds.

\section{Previous and related results}
\subsection{\PRAM model}
\par The multiple selection problem in \PRAM architecture has a solution matching given lower bound. Proposed algorithm works in $O(\log n)$ steps on $\frac{n}{\log n}$ processors \cite{optPRAMSelection}. The best known multiple selection algorithm needed $O(\frac{n}{p} \log (r+1))$ steps \cite{OptimalMultiselection}. However, algorithm proposed in that paper seems to work only for $p \in o(\frac{n}{\log n})$, thus for $r=1$, algorithm had optimal cost $O(\frac{n}{p})$ per processor, but required $\omega(\log n)$ steps of computation.

\subsection{BSP model}
\par For $\mathit{BSP}$  model the simple selection problem was investigated more often than the multiple selection version. There are two deterministic solution we would like to mention. The first is more or less parallel implementation of median of medians algorithm working in optimal time and $\min(\log p, \log \log n)$ synchronization rounds, where $p$ is a number of computational units and $n$ is the size of given set \cite{OptimalSelection}.
\par In 2010 Alexander Tiskin proposed a \TO algorithm solving the selection problem using regular sampling \cite{Tiskin2010}, which required $O(\log \log p)$ communication rounds.
\par If we allow randomization, constant communication round solution, given by Alexandros V. Gerbessiotis and Constantinos J. Siniolakis, is known for quite long time \cite{Gerbessiotis:1996:DSR:237502.237561}. 
\par The algorithm presented in this paper, also using regular sampling, needs only $O(\log_{(r+1)}^*n)$ communication rounds, thus for $r=1$ it gives \TO selection algorithm with round complexity $O(\log^*n)$, which is improvement over the previous result.
\section{Congested clique model}
In this paper we will focus on congested clique model. Congested clique is a variant of the congest model, where each two nodes are directly connected. Apart from the bandwidth size parameter and the number of rounds we will also take into account the time complexity of an algorithm, as without this additional limitation, we could simply sort set $S$ and $K$ using a constant round number algorithm \cite{LenzenRouting}, which would allow us pick elements with ranks in $K$ in some constant number of additional rounds.

\par Congested clique is a set of computational units $V$ which from clique - any two nodes can send message to each other. Each node of this clique has unique id from set $\{ 1,2,...,|V|\} $. Sometimes we will use $v$ or 'node' to address some element of set $V$, and we will denote id of $v$ by $v_{\mathit{id}}$. The selection of $k$-th element from set $S$ in this model is defined as follows:
\begin{itemize}
\item each node has up to $\frac{|S|}{|V|} + 1$ elements of $S$ in local memory as input
\item $S$ is not too big, $|S| \in \Theta(|V|^2)$
\item each node knows value of $k$
\item at the end some node must know value of $k$-th element of set $S$
\end{itemize}

\section{Median of median in parallel}
In the \RAM model to solve the selection problem we use the classical median of median algorithm. Basically it exploits the fact, that if we split a set into some pieces of size 5, select a median for each piece and then select a median of those medians, it split the original set into linearly large fractions. In parallel we could do similar thing, however instead of 5-element groups we will use elements stored in local memory of each node $v$ to form one group, thus the number of groups is equal to $|V|$ instead of $\frac{n}{5}$.
\par In \ref{MoM} algorithm and its analysis we will use additional denotations:
\begin{itemize}
  \item $S$ is a variable denoting set of elements, initially input set
  \item $S_v$ is avariable denoting part of $S$, stored in node $v$
  \item $c_i$ are some constant value
  \item $n$ is the size of the input set
  \item $S_i$ is a value of $S$ at the beginning of phase $i$
  \item $S_{(v,i)}$ is a value of $S_v$ at the beginning of phase $i$
\end{itemize}

Also at the beginning of each iteration of while loop (phase), we would like to have preserved the following invariant
$$ \forall{u,v\in V}\ | |S_{(v,i)}| - |S_{(u,i)}| | \leq 1$$

\begin{algorithm}
  \caption{Median of medians}
  \label{MoM}
  \begin{algorithmic}[1]
  \While{$|S|>c_0$}
    \State \textbf{each unit $v$ selects $y_v$ median of $S_v$}
    \State \textbf{each unit broadcasts it to other units}
    \State \textbf{each unit calculates $y$, median of received medians}
    \State \textbf{each unit $v$ calculates $x_v$, number of local elements smaller than $y$}
    \State \textbf{each unit broadcasts $x_v$ to other units}
    \State \textbf{each unit calculates $x$, the number of smaller than $y$ elements in $S$}
    \If{$k > x$}
    \State \textbf{for each unit $v$,  $S_v \gets \{z\in S_v\,|\,z>y\}$}
    \State \textbf{$k \gets k-x$}
    \Else
    \State \textbf{for each unit $v$,  $S_v \gets \{z\in S_v\,|\,z\leq y\}$}
    \EndIf
    \State \textbf{redistribute elements to ensure, that for each $v$, $S_v$ are of the same size}
    \State
    \Comment with accuracy of additive constant
  \EndWhile
  \end{algorithmic}
\end{algorithm}
\subsection{Number of rounds}
\par In $i$-th phase $y$ provides division of $S_i$ such that in $\{z\in S_i\,|\,z>y\}$ and in $\{z\in S_i\,|\,z\leq y\}$ we have at least $\frac{|S_i|}{4}$ elements. Thus, in each of them we have no more than  $\frac{3|S_i|}{4}$. Therefore, number of rounds, for $m$ elements is expressed as: $$R(m) \leq R\Big(\frac{3m}{4}\Big) + c_1$$ for some constant $c_1$. Thus, we can limit the size of $S_i$.
\begin{lemma}
\label{reductionlemma}
In $i$-th phase of \label{MoM} algorithm, we have subproblem of size $|S_i| \in O\big(n(\frac{3}{4})^k\big)$.
\end{lemma}
\par As a simple consequence we have, that $S_{c_2 \log n} \in O(1)$ for some constant $c_2$, therefore Algorithm \ref{MoM} requires $O(\log n)$ communication rounds.
\subsection{Redistribution of elements}
\par To preserve our invariant concerning even distribution of $S$ among nodes, we may have to redistribute elements of $S$. It can be done in two communication rounds. In the first each node announces to all other nodes value $x_v$, the number of elements $S_v$. Let $y_v = \sum\limits_{u \ | \ u_{\mathit{id}}<v_{\mathit{id}}} x_u$. Let assume, that node $v$ has in local memory set of elements $\{a_1, a_2, ..., a_{x_v}\}$. In the second round each node $v$ sends element $a_i$, to node with id $(i + y_v) \mod p$. Each node has no more than $|V|$ numbers, therefore  there would be at most one outgoing message per link. Also, whole set would be partitioned into subsets satisfying our invariant.
\subsection{Time}
\par Computation time $T(i)$ of the $i$-th phase is $O(\max_v |S_{(v,i)}|)$. Together with invariant that $S$ is evenly distributed among nodes and Lemma \ref{reductionlemma} gives us $T(i) \in O\big(\frac{|S_i|}{p}\big) = O\big(\frac{n}{p}(\frac{3}{4})^i\big)$. Thus overall time complexity is equal to $$\sum\limits_{i=1}^{R(n)} c_3\frac{|S_i|}{p} = \sum\limits_{i=1}^{R(n)} c_3\frac{n}{p}\Big(\frac{3}{4}\Big)^i \in O\bigg(\frac{n}{p}\bigg)$$
\par Summarizing, this algorithm provides optimal time, but require quite large number of communication rounds.

\section{Further parallelization}
In 2010 Alexander Tiskin proposed \cite{Tiskin2010} algorithm exploiting fact, that in the later phases of computation we could select regular sample instead of median in each node. 
\par A regular sample of size $h$ for set $S$ is a subset $S'$ of size $h$ with such property: if we consider $S$ and $S'$ in order, then between each two consecutive elements of $S'$ there are $O(\frac{|S|}{h})$ elements of set $S$. 
\par If we have a regular sample $H$ we can calculate rank of each element of $H$. Therefore, we can find two consecutive elemnts $x,y \in H$ such that $\mathit{rank}(x)\leq k \leq \mathit{rank}(y)$ thus we can reduce our problem to $(x,y) \cap S$ with $k' = k-rank(x)$. With sufficiently large regular sample, the size of set $S$ decreases significantly faster.
\par Alexander Tiskin proposed such sample sizes, that required number of communication rounds of his algorithm is $O(\log \log p)$,  still maintaining optimal time complexity $O(\frac{|S|}{p})$. In our result we describe how to select the size of samples slightly better to get an algorithm with even smaller number of rounds. Also, we analyze trade-off between time and round complexity. However, in the first place we must calculate how large regular sample can we compute in a constant number of rounds and time $O(\frac{n}{p})$.
\subsection{Distributed selection of regular sample}

\begin{lemma}
Let $H_v$ be a sample of set $S_v$ of size $h$.
\label{sampleSelection}
If we want to select sample $H$ of size $h$ for whole set $S = \bigcup\limits_v{S_v}$, we can select sample $H_v$ of size $h$ for each of sets $S_v$ and then select sample of size $h$ of set $H' = \bigcup\limits_v{H_{v}}$.
\end{lemma}
\begin{proof}
To prove this theorem we will give some upper bound on number of elements of set $S$ between two consecutive elements $x,y \in H$.
\par Let us focus for a while on one node $v$. Between two consecutive elements of $H_v$ there are at most $c_1\frac{|S|}{ph}$ elements of $S_v$. For $m$ elements of $H_v$ in interval $(x,y)$ we have at most $(m+1)c_1\frac{|S|}{ph}$ elements of $S_v$ in this interval. Note, that it is correct also for $m=0$. 
\par Thus, if some element of $H_v$ is in interval $(x,y)$ it contributes no more than $2c_1\frac{|S|}{ph}$ elements of $S_v$ to this interval. If there are no elements of $H_v$ in interval $(x,y)$, $S_v$ has no more than $c_1\frac{|S|}{ph}$ elements in interval $(x,y)$.
\par By definition we know, that for some $c$ there are at most $c\frac{ph}{h} = cp$ elements of $H'$ in interval $(x,y)$. 
\par Thus, we have at most $cp (2c_1\frac{|S|}{ph})$ elements from elements of $H'$ falling into $(x,y)$ and at most $pc_1\frac{|S|}{ph}$ elements from nodes $v$ such that none of $H_v$ is in $(x,y)$. There are at most $p$ such nodes, each contributes $c_1\frac{|S|}{ph}$ to this interval. Together it gives us no more than $(2cc_1 + c_1) (p\frac{|S|}{ph}) \in O(\frac{|S|}{h})$ elements in interval $(x,y)$. Thus, if we want to select regular sample for set $S$ we may select $H_v$ for each of sets $S_v$, send them to one node and there select regular sample of $H'$. Therefore Lemma \ref{sampleSelection} is correct. 
\end{proof}
\par If we have a set of elements of size $n$ and we want to select $h$ elements of known ranks we need \cite{optMSelection} $\Omega(n \log h)$ operations. Also with slightly modified quick sort algorithm it is possible to select those elements in time $O(n \log h)$. Thus $\Theta(n \log h)$ time is sufficient and required to perform local sample selection. 
\par If we solve equation $\frac{|S_i|}{p} \log h = \frac{n}{p}$ for $h$, we will get $h = 2^{(\frac{n}{|S_i|})}$. Thus, in $i$-th round we can select local sample of size $2^{(\frac{n}{|S_i|})}$ and we will still hold our limit of $O\big(\frac{n}{p}\big)$ time for each node in each round.

\subsection{Inefficient solution}
Presented algorithm in each \while loop iteration reduces the size of problem. It is done, as mentioned before, by calculating ranks of elements in regular sample and finding interval containing $k$-th element. Such interval will be called \textit{active}.

\begin{algorithm}
  \caption{Selection by regular sampling}
  \label{select}
  \begin{algorithmic}[1]
  \While{$|S|>c\frac{n}{log n}$}
    \State \textbf{each unit $v$ selects regular sample $Y_v$ of size  $2^{(\frac{n}{|S|})}$ from $S_v$}
    \State \textbf{sort $Y'=\bigcup\limits_{v\in V} Y_v$ using distributed algorithm}
    \State
    \Comment sorting requires constant number of rounds \cite{LenzenRouting}
    \State \textbf{calculate $Y$, regular sample for set $Y'$ of size $2^{(\frac{n}{|S|})}$}
    \State
    \Comment which is also regular sample of $S$
    \State \textbf{each unit with some element from $Y$ in local memory broadcast it to all other nodes}
    \State
    \Comment all nodes have Y in local memory
    \State \textbf{each unit $v$ calculates values $x_{(v,j}$, number of $S_v$ elements in $j$-th interval induced by $Y$}
    \State 
    \Comment by intervals induced by $Y$ we understand intervals with bounds in two consecutive elements of $Y$
    \State \textbf{each unit $v$ sends $x_{(v,j)}$ to node with $id = j$, for each $j$ }
    \State \textbf{each node aggregates received values and broadcasts sum it to all nodes}
    \State 
    \Comment each unit knows $X$, the set numbers of elements of $S$ in each interval induced by $Y$
    \State \textbf{each unit calculates active interval $\langle x,y \rangle$}
    \State \textbf{for each unit $v$ $S_v \gets \{z \,|\,z\in S_v \wedge x \leq z < y \}$}
    \State \textbf{recalculate value of $k$ with respect to active interval}
    \State \textbf{redistribute elements to ensure, that for each $v$ $S_v$ are of the same size}
    \State
    \Comment with accuracy of additive constant
  \EndWhile
  \State sort $S$
  \State
  \Comment again - sorting requires constant number of rounds \cite{LenzenRouting}
  \end{algorithmic}
  
\end{algorithm}

\subsubsection{Analysis of Algorithm \ref{select}}
\label{analysis}
\begin{lemma}
\label{sampleSizeLemma}
Size of sample is $O(\frac{n}{p \log n})$ 
\end{lemma}
\begin{proof}
By contradiction. If sample would be of size $\omega(\frac{n}{p \log n})$, that would mean $2^{(\frac{n}{|S_i|})} > \frac{n}{p \log n}$. Thus $|S_i|< \frac{n }{\log(\frac{n}{p \log n})} < c \frac{n }{\log n}$, thus $S_i$ would be to small to allow another execution of \while loop.
Therefore, Lemma \ref{sampleSizeLemma} is correct.
\end{proof}
\begin{lemma}
\label{linearTimeLemma}
Each iteration of \while loop takes $O(\frac{n}{p})$ time.
\end{lemma}

\begin{proof}
\par \ 
\begin{itemize}
\item  we used the size of sample resulting in exactly $\Theta(\frac{n}{p})$ required time for local sample selection
\item calculating number of elements of $S_{(v,i)}$ in each interval induced by $Y$ require up to $\Theta({\frac{|S_i|}{p}\log 2^\frac{n}{S_i}}) = \Theta(\frac{n}{p})$ time, as we must execute binary search on sample of size $2^\frac{n}{|S_i|}$ for $|S_{(v,i)}|$ elements. 
\item sorting $Y'$ takes $(\frac{n}{p})$, as by Lemma \ref{sampleSizeLemma} one local sample is of size $O(\frac{n}{p \log n})$, thus $|Y'| \in O(\frac{n}{\log n})$
\item other operations plainly are executable in $O(\frac{n}{p})$ time.
\end{itemize}

Thus every operation inside \while loop iteration takes $O(\frac{n}{p})$ time and Lemma \ref{linearTimeLemma} is correct.
\end{proof}

Now let us calculate the number of rounds of this algorithm. In the $i$-th round we select sample of size $2^{(\frac{n}{|S_i|})}$. We can locally select such sample thus, by Lemma \ref{sampleSelection} we can do it for whole $S_i$. To make analysis slightly easier, let us define sequence $g_i = \frac{n}{|S_i|}$. Now we must find $k$, such that $g_k \in \Theta(n)$. Such $k$ would imply, that $|S_k|$ is $O(1)$, thus we reduced problem to a subproblem of the constant size and it is possible to solve it locally on one machine.
\par As we mentioned in $i$-th phase we select regular sample of size $2^{g_i}$. Thus $|S_{i+1}| \in \Theta({\frac{|S_i|}{2^{g_i}}})  \Rightarrow  2^{g_i} \in O(\frac{|S_i|}{|S_{i+1}|}) = O(\frac{n}{g_i|S_{i+1}|}) = O(\frac{g_{i+1}}{g_i}) \Rightarrow 2^{g_i} \leq g_{i+1}$, as $g_i$ is larger than $c$ behind big $O$ notation. Let us consider another sequence $f_i$ defined as follows:
$$f_0 = 1;\,f_i=2^{f_{i-1}} \text{ for } i>0$$
Sequence $g_i$ grows not slower than $f_i$, thus if we would have $f_k > n$ that would imply $g_k > n$. This simple analysis shows us, that with sample size for $i$-th phase equal to $2^{g_i}$ we would have $|S_{\log^*n}| \in O(1)$.

\par Summarizing, such sample sizes give us deterministic $O(\log^*n)$ round algorithm. With linear time complexity per round \ref{linearTimeLemma} this algorithm unfortunately requires $\Theta(\frac{n\log^*n}{p})$ time, thus it is not an  algorithm. 

\subsection{Preprocessing and \TO solution}
There are two possible ways to improve presented algorithm. We can either modify slightly function $g_i$ to get linear time solution or provide preprocessing reducing the size of problem to $\frac{n}{\log* n}$. First solution requires analyzing some quite complicated equations, so we will present second one - provide preprocessing and use presented inefficient algorithm as a black-box.
\begin{theorem}
\label{selectionTheorem}
It is possible to solve the selection problem in $O(\log^* n)$ communication rounds with deterministic, \TO algorithm.
\end{theorem}

\begin{proof} Sufficient algorithm for our preprocessing was almost presented at the beginning \ref{MoM}. Only difference is, that we run our \while loop $\Theta(\log \log^* n)$ times. By Lemma \ref{reductionlemma}, we will decrease the size of $S$ to required $O(\frac{n}{\log^*n})$. This preprocessing will requires $\Theta(\log \log^* n) \subseteq O(\log^*n)$ communication rounds and $O(\frac{n}{p})$ time. Inefficient black-box algorithm requires $O(\log^* \frac{n}{\log^*n}) \subseteq O(\log^* n)$ communication rounds and $O(\frac{\frac{n}{\log^*n}\log^*(\frac{n}{\log^*n})}{p}) \subseteq O(\frac{n}{p})$ time. As both part require $O(\log^* n)$ communication rounds and $O(\frac{n}{p})$ time, whole algorithm does, thus Theorem \ref{selectionTheorem} is correct.
\end{proof}
\subsection{Time - number of rounds - trade-off}
Provided algorithm is \TO but unfortunately does not work in constant round number. It seems to be natural to calculate how inefficient must be our algorithm to perform calculations in some constant number of steps. Let us denote by $T_A$ time required by algorithm $A$ to perform selection.
Now we can define $\phi_A$, an inefficiency factor of algorithm $A$ as $\phi = \frac{T_A}{\frac{n}{p}} \Rightarrow \phi n = pT_A$. Now let us repeat whole analysis from \ref{analysis} with limit for our work set to $pT_A$ instead of $n$. 
In each round we can select sample of size $2^{\frac{pT_A}{|S_i|}} = 2^{\phi \frac{n}{|S_i|}} = ({2^\phi})^\frac{n}{|S_i|}$. Thus, function $g_i$ is larger than $f'_i$ defined as follows: $$f'_0=1; f'_i = (2^\phi)^{f'_{i-1}} \text{ for } i>0$$ This lower bound on value of $g_i$ give us round complexity $O(\log^*_{2^\phi} n)$. Therefore if we choose an inefficiency factor $\phi$ such that $log^*_{2^\phi} n \in O(1)$ we will get a constant round number solution.

\section{Multiple selection by regular sampling}
\subsection{Motivation}
For the sorting problem we know a deterministic \TO constant round number solution, but for the selection problem we do not. Both problems are actually edge cases of more general multiple selection problem. In such problem for a given set $S$ and set of ranks $K$ we want to know elements of $S$ with rank present in set $K$. For $|K| = 1$ it becomes the selection problem and for $|K|= |S|$ it becomes the sorting problem. It seems to be interesting, what happens for $1<|K|<n$. 

\subsection{Problem}
In multiple selection problem we have a set of ranks $K$ instead of single rank $k$. We assume that set $K$ is small, as if $|K| > \frac{|V|}{p}$ then we could simply sort $S$ and $K$ \cite{LenzenRouting} and select from $S$ elements with rank in $K$ and that would be constant communication round, \TO solution. As $K$ is small, we can also assume, that each node has it in local memory.

\subsection{Algorithm}
Algorithm \ref{mselect} is similar to the algorithm for selection by regular sampling presented before. Main differences are, that we have more than one active interval and we are allowed to work for $O(\frac{n}{p} \log|K|)$ time.
\subsection{Preprocessing}
This causes some changes in both preprocessing and main part of the algorithm. As for preprocessing part, instead using median of medians, we must use a regular sample of size $\Theta(|K|^2)$. We may have up to $|K|$ active intervals, and we want decrease the size of problem $\Theta(|K|)$-times in one phase of preprocessing, thus the sample must have size $\Theta(|K|^2)$.
\par As for the main part of algorithm, we are allowed to select slightly larger regular samples still having \TO algorithm.
\begin{algorithm}
  \caption{Multiple selection by regular sampling}
  \label{mselect}
  \begin{algorithmic}[1]
  \While{$|S|>c\frac{n}{\log n}$}
    \State \textbf{each unit $v$ selects regular sample $Y_v$ of size  $(|K|+1)^{(\frac{n}{|S|})}$ from $S_v$}
    \State \textbf{sort $Y'=\bigcup\limits_{v\in V} Y_v$ using distributed algorithm}
    \State
    \Comment sorting requires constant number of rounds \cite{LenzenRouting}
    \State \textbf{calculate $Y$, regular sample for set $Y'$ of size $2^{(\frac{n}{|S|})}$}
    \State
    \Comment which is also regular sample of $S$
    \State \textbf{each unit with some element from $Y$ in local memory broadcast it to all other nodes}
    \State
    \Comment all nodes have $Y$ in local memory
    \State \textbf{each unit $v$ calculates values $x_{(v,j)}$, number of $S_v$ elements in $j$-th interval induced by $Y$}
    \State 
    \Comment by intervals induced by $Y$ we understand intervals with bounds in two consecutive elements of $Y$
    \State \textbf{each unit $v$ sends $x_{(v,j}$ to node with $id = j$, for each $j$ }
    \State \textbf{each node aggregates received values and broadcasts sum it to all nodes}
    \State 
    \Comment each unit knows $X$, the set numbers of elements of $S$ in each interval induced by $Y$
    \State \textbf{each unit calculates the set of active intervals $I$}
    \State \textbf{for each unit $v$ $S_v \gets \{z \,|\,z\in S_v \wedge \exists \langle x,y \rangle \in I . x \leq z < y \}$}
    \State \textbf{recalculate values of $K$ with respect to active intervals}
    \State \textbf{redistribute elements to ensure, that for each $v$ $S_v$ are of the same size}
    \State
    \Comment with accuracy of additive constant
  \EndWhile
  \State sort $S$
  \State find elements with rank in $K$
  \State
  \Comment again - sorting requires constant number of rounds \cite{LenzenRouting}
  \end{algorithmic}
\end{algorithm}

\subsection{Analysis of Algorithm \ref{mselect}}
Analysis is quite similar as for \ref{select} Selection by regular sampling algorithm. Difference is, that we are allowed to select sample of size $2^{(\frac{n}{S_i} \log (|K|+1))} = (|K|+1)^{\frac{n}{|S_i|}}$, as optimal sequential time is $\Theta(n \log|K|)$ \cite{optMSelection}. Also, instead of one active interval we have up to $|K|$. Therefore, our dependency between $g_{i}$ and $g_{i+1}$ is slightly different.
$$|S_{i+1}| \in O\bigg(|K|\frac{|S_i|}{(|K|+1)^{\frac{n}{|S_i|}}}\bigg)$$ 
$$(|K|+1)^{({g_i}-1)} \in O\Big(\frac{g_{i+1}}{g_i}\Big)$$
$$(|K|+1)^{({g_i}-1)} \leq g_{i+1}$$
and is limited from below by $f_i$ defined as follows.
$$
f_0 = |K|+1 ;f_i = (|K|+1)^{(f_{i-1}-1)} \text{ for } i>0
$$
Therefore $|S_{c \log^*_{(|K|+1)} n}| \in O(1)$ for some $c$, which means our algorithm requires  $O(\log^*_{|K|+1} n)$ communication rounds. That mean, for $K$ satisfying inequality $(|K|+1)\upuparrows{\epsilon}$\footnote[2]{Knuth's up-arrow notation}$ > n$, for some $\epsilon \in O(1)$, our algorithm requires constant time of communication rounds.

\section{Summary}
In this paper we presented a deterministic, \TO multiple selection algorithm using regular sampling technique. Size of sample seems to be as large as possible, without violating time restriction, which suggest, that for better result we need use some additional techniques. We still do not know whether there are deterministic, optimal time selection algorithms requiring constant number of communication rounds. 

\par As mentioned ad the beginning, presented algorithm works also in slightly different models. Congested clique algorithms may be easily implemented in $\mathit{BSP}$. As for MapReduce, we could use result presented by James Hegeman and Sriram V. Pemmaraju \cite{ Hegeman2014}, and get some robust algorthm, as solution presented in this paper is communication and memory efficient. 
\par Actually algorithm in those other models works even if we change assumption, that $n \in \Theta(p^2)$ into $p \in \Theta(n^{1-\epsilon})$ for some constant $\epsilon$. However, this requires some additional effort while restoring the invariant concerning the equal distribution of data among machines.

\bibliographystyle{unsrtnat} 

\bibliography{Bibliography} 

\end{document}